\pdfoutput=1
\documentclass[10pt,pra,aps,superscriptaddress,nofootinbib,twocolumn,longbibliography]{revtex4-1}
\usepackage{amsmath,bbm}
\usepackage{amsthm}
\usepackage{amsmath}
\usepackage{latexsym}
\usepackage{amssymb}
\usepackage{float}
\usepackage{graphicx}           
\usepackage{color}
\usepackage{xcolor}
\usepackage{mathpazo}
\usepackage{comment}
\usepackage{enumitem}
\usepackage{multirow}
\usepackage{setspace}
\usepackage[colorlinks=true,linkcolor=blue,citecolor=magenta,urlcolor=blue]{hyperref}
\usepackage{svg}
\usepackage{cancel}
\usepackage{physics}

\newcommand{\be}{\begin{equation}}
\newcommand{\ee}{\end{equation}}
\newcommand{\bea}{\begin{eqnarray}}
\newcommand{\eea}{\end{eqnarray}}

\def\squareforqed{\hbox{\rlap{$\sqcap$}$\sqcup$}}
\def\qed{\ifmmode\squareforqed\else{\unskip\nobreak\hfil
\penalty50\hskip1em\null\nobreak\hfil\squareforqed
\parfillskip=0pt\finalhyphendemerits=0\endgraf}\fi}
\def\endenv{\ifmmode\;\else{\unskip\nobreak\hfil
\penalty50\hskip1em\null\nobreak\hfil\;
\parfillskip=0pt\finalhyphendemerits=0\endgraf}\fi}
\newcommand{\I}{\mathbbm{1}}

\newcommand{\la}{\langle}
\newcommand{\ra}{\rangle}

\newcommand{\re}{\color{blue}}  
\newcommand{\blk}{\color{black}}

\makeatletter
\newtheorem*{rep@theorem}{\rep@title}
\newcommand{\newreptheorem}[2]{%
\newenvironment{rep#1}[1]{%
 \def\rep@title{#2 \ref{##1}}%
 \begin{rep@theorem}}%
 {\end{rep@theorem}}}
\makeatother

\newtheorem{thm}{Theorem}
\newreptheorem{thm}{Theorem}

\newtheorem{obs}{Observation}

\begin{document}


\title{Unbounded quantum advantage in communication complexity measured by distinguishability}


\author{Satyaki Manna}
\affiliation{School of Physics, Indian Institute of Science Education and Research Thiruvananthapuram, Kerala 695551, India}
\author{Anubhav Chaturvedi}
\affiliation{Faculty of Applied Physics and Mathematics,
 Gda{\'n}sk University of Technology, Gabriela Narutowicza 11/12, 80-233 Gda{\'n}sk, Poland}
\affiliation{International Centre for Theory of Quantum Technologies (ICTQT), University of Gda{\'n}sk, 80-308 Gda\'nsk, Poland}
\author{Debashis Saha}
\affiliation{School of Physics, Indian Institute of Science Education and Research Thiruvananthapuram, Kerala 695551, India}


\begin{abstract}
Communication complexity is a fundamental aspect of information science, concerned with the amount of communication required to solve a problem distributed among multiple parties. 
The standard quantification of one-way communication complexity relies on the minimal dimension of the communicated systems. In this paper, we measure the communication complexity of a task by the minimal \emph{distinguishability} required to accomplish it, while leaving the dimension of the communicated systems unconstrained. {Distinguishability} is defined as the maximum probability of correctly guessing the sender's input from the message, quantifying the message's distinctiveness relative to the sender's input. This measure becomes especially relevant when maintaining the confidentiality of the sender's input is essential. After establishing the generic framework, we focus on three relevant families of communication complexity tasks---the random access codes, equality problems defined by graphs and the pair-distinguishability tasks. We derive general lower bounds on the minimal classical distinguishability as a function of the success metric of these tasks. We demonstrate that quantum communication outperforms classical communication presenting explicit protocols and utilizing semi-definite programming methods. In particular, we demonstrate \emph{unbounded quantum advantage} for random access codes and Hadamard graph-based equality problems. Specifically, we show that the classical-to-quantum ratio of minimal distinguishability required to achieve the same success metric escalates polynomially and exponentially with the complexity of these tasks, reaching arbitrarily large values. 
\end{abstract}

\maketitle


\section{Introduction} 

Communication complexity is a cornerstone of information science, finding applications in diverse fields such as distributed computation, query complexity, game theory, integrated circuit design, cryptography, and streaming algorithms \cite{ccbook,TRbook,rao_yehudayoff_2020,Yao}. In its most fundamental form, communication tasks involve two entities—a sender and a receiver—aiming to compute a function based on both of their inputs. The traditional measure for one-way communication complexity pertains to the minimal dimension of the systems the sender communicates to accomplish the task. Quantum theory provides an unbounded advantage over classical communication protocols in several communication complexity problems. 
In particular, the ratio between the minimum dimensions of classical and quantum messages required to accomplish a given task can increase to arbitrarily large values  \cite{Raz,RevModPhys.82.665,quantumFP,buhrman1998quantum,DEWOLF,pnas,saha2019,saha2023}.

In this study, instead of minimal dimension, we take \emph{minimal distinguishability} as the measure of communication complexity. 
Distinguishability is the maximum average probability of distinguishing (guessing) the sender's input from the communicated message, quantifying the extent to which information about the sender's input data is revealed from the transmission. In particular, distinguishability is also called guessing probability information and corresponds to a one-shot version of accessible information based on min-entropies \cite{Tavakoli2022informationally,InformationMeasure}. Therefore, a higher required distinguishability corresponds to a higher amount of information revealed about the sender's input and a higher amount of necessary communication. Conversely, a lower minimum distinguishability indicates a lower amount of necessary communication. 

Besides being a discrete measure, the dimension as a measure of communication only provides a partial characterization. On the other hand, distinguishability is a continuous measure of communication and is better suited for providing a complete picture. 
Additionally, accomplishing a communication task while maintaining minimal distinguishability becomes crucial when the confidentiality of the sender's input is a concern. Traditional communication complexity, measured by the minimum dimension of the communicated systems, overlooks the privacy of the sender's input. On the other hand, cryptography prioritizes communication privacy but does not inherently consider the complexity of the communication process \cite{Ekert,QCryp}. Hence, our approach addresses both aspects within a unified framework. Moreover, as distinguishability corresponds to the information content in the pre-measurement ensemble, the framework remains independent of the particulars of the type quantum or classical communication protocols. Finally, unlike minimum dimension, distinguishability is critically relevant to the study of ontological models of operational theories \cite{bod,PhysRevA.109.032212}.

To commence, we discuss the generic framework of ''communication complexity based on distinguishability''. This framework is related to the framework for communication tasks with bounded distinguishability introduced in Refs. \cite{bod,Tavakoli2020informationally}, further developed in Refs. \cite{Tavakoli2022informationally,pauwels2024informationcapacityquantumcommunication}.
We outline methodologies for determining the minimal distinguishability required to achieve a given value of the success metric associated with communication complexity problems for classical and quantum communication protocols. To quantify the quantum advantage, we take the ratio between the minimum required distinguishability for classical and quantum communication to achieve a specific value of the success metric of the task. 

Moving to our main results, we examine three classes of communication complexity tasks with diverse applications. The first class involves the random access codes (RAC), where the sender possesses a string of $d$its of size $n$, and the receiver seeks to guess a randomly chosen $d$it from that string \cite{RAC1,RAC,RACA}. The second category is an equality problem defined by graphs, where the objective is to determine whether the sender's and receiver's inputs are the same or not, given that the inputs randomly occur according to a graph \cite{saha2023}. The third class is the pair-distinguishability task, a generalized version of the task introduced in \cite{bod}. We analytically establish generically applicable lower bounds on the distinguishability as a function of the success metric and the specifications of the task [i.e., $(d,n)$ for RAC, the properties of the graph for the equality problem, and $n$ for the pair-distinguishability tasks], for classical communication protocols. We then demonstrate several instances of quantum advantage in these tasks, either by presenting explicit quantum protocols or by employing semi-definite programming methods, i.e., we show that quantum protocols, in general, require less distinguishability than classical communication protocols to achieve the same success metric in these tasks. We also provide a comprehensive study of quantum advantage for the RAC with $d,n=2,3$ and for equality problems based on odd-cycle graphs. 

Among our diverse findings, the demonstration of an unbounded quantum advantage in two families of tasks stands out. Specifically, we show that the ratio between distinguishability in classical and quantum communication to attain the same success metric increases with the level $d$ of Alice's input for the RACs, and similarly, for Hadamard graphs \cite{FR}, this ratio grows exponentially with the size of the graph. 
 Furthermore, we present quantum protocols that offers an advantage in terms of distinguishability but not with respect to traditional dimensional quantification for the same task. 
Finally, we summarize insights gained into communication complexity measured by distinguishability and enlist future research directions.

\section{communication complexity based on distinguishability}\label{II}

First, let us establish the notation $[K]$ to represent the set $\{1, \cdots, K\}$ for any positive integer $K$. In a one-way communication complexity task, Alice, the sender, is assigned an input variable $x$ sampled from the set $[N]$. The probability of obtaining $x$ is denoted by $p_x$, with $p_x = 1/N$ for a uniform distribution. In each round of the task, based on the value of $x$, Alice transmits a message (either classical or quantum) to Bob, the receiver. Bob is also given an input variable $y \in [M]$, and depending on $y$ and the received message, he produces an output $z \in [D]$. Multiple rounds of this task are performed to collect frequency statistics represented by conditional probabilities $p(z|x, y)$. The objective is to maximize a success metric of the form:
\be \label{SGen}
\mathcal{S} = \sum_{x,y,z} c(x,y,z) p(z|x,y) ,
\ee
where $c(x,y,z)\geqslant 0$, and the metric is normalized by ensuring $\sum_{x, y, z} c(x, y, z) = 1$ so that the maximum value of $\mathcal{S}$ is 1. 

The communication is subject to a constraint on the distinguishability of inputs $x$. The distinguishability is defined as the maximum average guessing probability of determining input $x$ from the message (classical or quantum) using the optimal measurements:
\be \label{pD1}
\mathcal{D} = \max_{M}\Bigg\{ \sum_{x} p_x p(z=x|x,M)\Bigg\},
\ee
where the maximization is over all possible measurements $M$ in the theory. Allowing $\mathcal{D}$ to be 1 would enable Alice to simply transmit the input $x$, say via classical $N$-dimensional systems, making the task trivial. The task becomes nontrivial when $\mathcal{D}$ is less than 1. It is important to note that the success metric of communication complexity task is typically associated with guessing some function of $x$ and $y$, i.e.,
\begin{equation} \label{Sgf}
\mathcal{S} = \sum_{x, y} c(x, y) p(z = f(x, y) | x, y).
\end{equation}
 We note here that the success metric \eqref{Sgf} is a particular instance of the general success metric given in \eqref{SGen}, wherein the output $z$ should be a specific function of inputs $x$ and $y$. In this paper, we focus on family of tasks that have a success metric of the form \eqref{Sgf}. We will now find how such a communication task translates into classical and quantum communication protocols, and what is meant by the term \emph{quantum advantage} in this context.

\subsection{Classical communication}

While communicating classically, Alice sends a $d$-labelled message $m\in [d]$ depending on the input $x$. Any encoding strategy is described by probability distributions of sending $m$ given input $x$, $\{p_e(m|x)\}$, where $\sum_m p_e(m|x)=1$ for all $x$. Similarly, the generic decoding strategy for providing Bob's output is defined by the probability distributions over output $z$ given $m,y$ denoted as $\{p_d(z|y,m)\}$, where $\sum_z p_d(z|y,m)=1$ for all $m,y$. It is important to note that there is no restriction on the dimension $d$, implying $m$ can take an arbitrarily large number of distinct values. Using this fact, it can be shown, as detailed in \cite{bod} (Observation 4), that sharing classical randomness is not beneficial for this task. For an generic classical encoding and decoding, the resulting conditional probability is expressed as follows:
\be \label{pC}
p(z|x,y) = \sum_m\sum_{x,y} p_e(m|x) p_d(z|y,m) .
\ee 
Let us now obtain the expression of distinguishability given an encoding. By substituting the probabilities \eqref{pC} into \eqref{pD1} and leveraging the property that the maximum value of any convex combination of a set of numbers is the maximum number from that set, we obtain
\bea 
\mathcal{D}_C &=& \max_{\{p_d(z|m)\}}\sum_m\sum_{x} p_x p_e(m|x) p_d(z=x|m)\nonumber\\
              &=& \max_{\{p_d(z|m)\}} \sum_m \left(\sum_x p_x p_e(m|x) p_d(z=x|m) \right)\nonumber\\
              &=& \sum_m\max_{x}\bigg\{ p_x p_e(m|x)\bigg\}.
\eea
For the uniform distribution $p_x=1/N$, the above simplifies to
\be \label{DC}
\mathcal{D}_C =
\frac{1}{N}\sum_m\max_{x}\bigg\{ p_e(m|x)\bigg\}.
\ee 
It turns out that the expression of the optimal value of the success metric \eqref{SGen} in classical communication, denoted by $\mathcal{S}_C,$ can be expressed only in terms of encoding $\{p_e(m|x)\}$. 
By substituting \eqref{pC} into \eqref{SGen} and using the fact that the maximum value of any convex combination of a set of numbers is the maximum number from that set, we find
\begin{widetext}
\bea \label{scd}
\mathcal{S}_C &=& \max_{\{p_e(m|x)\},\{p_d(z|y,m)\} } \sum_{x,y,z} \sum_m c(x,y,z)  p_e(m|x) p_d(z|y,m) \nonumber \\
&=& \max_{\{p_e(m|x)\},\{p_d(z|y,m)\} } \sum_{y,z} \sum_m \left( \sum_x c(x,y,z)  p_e(m|x)  \right) p_d(z|y,m)\nonumber \\
&=& \max_{\{p_e(m|x)\} } \sum_{y,m} \max \left\{ \left( \sum_x c(x,y,z=1)  p_e(m|x)  \right), \cdots, \left( \sum_x c(x,y,z=D)  p_e(m|x)  \right) \right\}.
\eea 
For the particular case of binary-outcome, $z\in [2]$, the above expression simplifies to 
\bea \label{scdz2}
\mathcal{S}_C
&=& \max_{\{p_e(m|x)\}} \left[ 1 - \sum_{m,y} \min \left\{ \sum_{x} c(x,y,z=1)p_e(m|x), \sum_{x} c(x,y,z=2)p_e(m|x) \right\}  \right] .
\eea 
\end{widetext}

\subsection{Quantum communication}

In quantum communication, Alice transmits a quantum state $\rho_x$ acting on $\mathbbm{C}^d$, and Bob's output is the result of a quantum measurement described by sets of positive semi-definite operators $\{M_{z|y}\}_{z,y}$, with $z\in [D], y \in [M]$, and $\sum_z M_{z|y}=\I $. It is worth noting that the dimension of quantum systems, denoted as $d$, can take arbitrary values. This setup results in the following statistics:
\be \label{pD}
p(z|x,y) = \tr (\rho_x M_{z|y}) .
\ee 
For any such quantum strategy we denote the resultant success metric \eqref{SGen} in a generic communication complexity task as,
\be \label{SGenQ}
\mathcal{S}_Q = \sum_{x,y,z} c(x,y,z) \tr (\rho_x M_{z|y}),
\ee

The distinguishability of the sender's quantum states $\{\rho_x\}$ is given by,
\be \label{DistinguishabilityQuantum}
\begin{split}
 \mathcal{D}_Q=\underset{\{M_x\}}{\ \ \max\ \ } & \sum_x p_x \tr(\rho_x M_{x}))\\
\text{s.t.\ \ } &  M_x \geqslant 0, \ \ \forall x \in [n]\\
&  \sum_x M_x = \I, 
\end{split}
\ee 
which forms a straightforward semi-definite program. Let us now make the following observation, which will come in handy later.
\begin{obs} \label{fact:d/N}
The distinguishability of $N$ quantum states that belong to  $\mathbbm{C}^d$ sampled form uniform distribution $p_x=1/N$,
    \be \label{Dd/N}
    \mathcal{D}_Q \leqslant \frac{d}{N}.
    \ee
\end{obs}
\begin{proof} 
For uniform distribution, the distinguishability of a given set of quantum states is given by 
\be 
\mathcal{D}_Q = \max_{\{M_x\}} \frac1N \sum_x \tr(\rho_x M_{x})  \leqslant \frac{1}{N} \sum_{x} \tr( M_{x})  = \frac{d}{N} .
\ee 
Here, we use the fact that the $\tr (\rho_x) =1$, $M_{x}\geqslant 0, \sum_x M_{x}= \I$.
\end{proof}    

In Appendix \ref{SDP}, we present easy-to-implement semi-definite programming (SDP) techniques to retrieve the maximum quantum value of the success metric given an upper bound on distinguishability based on the methods originally formulated in \cite{Chaturvedi2021characterising,Tavakoli2022informationally}.

\subsection{Quantifying quantum advantage}

In this work, our objective is to measure the quantum advantage in communication complexity with respect to the distinguishability of $x$. Quantum advantage is established when, to achieve a given value of a success metric $\mathcal{S}$, the minimum value of classical distinguishability $\mathcal{D}_C$ surpasses the value of quantum distinguishability $\mathcal{D}_Q$. To quantify this, we draw an analogy from the standard notion of communication complexity and consider the ratio of the distinguishability in classical and quantum communication, 
$\mathcal{D}^\mathcal{S}_C / \mathcal{D}^\mathcal{S}_Q $, given that the success metric attains at least a certain value, $\mathcal{S}$, in both, classical and quantum communication. If this ratio exceeds 1, it indicates a quantum advantage. In particular, we aim to derive a relationship of the following form:
\be 
\mathcal{F} (\mathcal{S}_C) \leqslant \mathcal{D}_C
\ee 
where $\mathcal{F}(\mathcal{S}_C)$ is some function on $\mathcal{S}_C$ and specifications of the task. This enables us to establish a lower bound on $\mathcal{D}_C$ given a value of $\mathcal{S}$. Subsequently, we present quantum communication protocols achieving the same value $\mathcal{S}$ such that $\mathcal{D}_Q$ is less than the obtained lower bound on $\mathcal{D}_C$. The advantage is considered unbounded if the ratio $\mathcal{D}^\mathcal{S}_C / \mathcal{D}^\mathcal{S}_Q $ can become arbitrarily large.

Additionally, the quantum advantage can be quantified in the reverse direction by determining the ratio between the success metric, $\mathcal{S}^\mathcal{D}_Q / \mathcal{S}^\mathcal{D}_C$, under the condition that the distinguishability cannot be greater than the value, $\mathcal{D}$, in both, quantum and classical communication.

\section{Random access codes}\label{III}

In this task, Alice is provided with a string of $n$ dits, denoted as $x=x_1x_2\cdots x_n$, uniformly randomly selected from the set encompassing all possible strings. Each $x_y$ in the string belongs to the set $[d]$ for all $y \in [n]$. During communication, Alice conveys information about the acquired string using either classical or quantum systems. Bob's task is to deduce the $y$-th dit, with $y$ being randomly chosen from the set $[n]$. The success metric is determined by the average success probability, defined through 
\be \label{Snd}
\mathcal{S}(n,d) = \frac{1}{nd^n} \sum_{x,y} p(z=x_y|x,y) .
\ee
Here the task is uniquely defined by the values of $n$ and $d$. Utilizing \re\eqref{DC}\blk, we can express the distinguishability for any specific encoding as follows:
\be\label{scd2}
\mathcal{D}_C=\frac{1}{d^n}\sum_m \max_{x}\bigg\{ p_e(m|x)\bigg\}.
\ee

We now demonstrate a generically applicable lower bound on the minimum distinguishability required to achieve a given value of the success metric for $(n,d)$ RACs,

 \begin{thm}
     For any $n$ and $d$, the following holds in classical communication
     \be\label{scnd} 
n \mathcal{S}_{C}(n,d) +1-n \leqslant \mathcal{D}_C.
     \ee 
 \end{thm}
 \begin{proof}
First, we express this average success metric \eqref{Snd} in terms of the general form of success metric \eqref{SGen} as
\be 
    c(x,y,z)= \frac{1}{nd^n} \times
 \begin{cases}
   1, & \text{if $z=x_y$}\\
    $0$,             & \text{otherwise.} 
\end{cases}
\ee 
Substituting this expression of $c(x,y,z)$ into \eqref{scd}, we get,
 \begin{widetext}
     \be\label{cc1}
    \mathcal{S}_C(n,d) = \max_{p_e(m|x)} \Bigg[\frac{1}{nd^n} \sum\limits_{m} \Bigg\{ \underbrace{ \sum^n_{y=1} \max \bigg\{\sum_{x|x_y=1} p_e(m|x) ,\cdots, \sum_{x|x_y=d} p_e(m|x)\bigg\} }_{\chi(m)} \Bigg\} \Bigg] .
    \ee 
\end{widetext}
The underlined expression in \eqref{cc1} is denoted by $\chi(m)$. Here, the notation $x|x_y=d$ in the summation subscript indicates that the summation is taken over all $x_i$, \emph{except}  $x_y$, which is set to $d$. The next step is to show that the following relation holds for every $m$;
    \be \label{bc}
    \chi(m) \leqslant (n-1)\sum_x p_e(m|x) +\max_x p_e(m|x).
    \ee   
Let us fix a particular value of $m$. Given $m$, for every $y$, we denote the value of $x_y$ by $d^*_y$ for which 
\be 
\sum_{x|x_y=d_y^*} p_e(m|x) \geqslant \sum_{x|x_y} p_e(m|x), \quad \forall x_y .
\ee 
In other words, for every $y$,
\bea  
    && \max \bigg\{\sum_{x|x_y=1} p_e(m|x), \cdots, \sum_{x|x_y=d} p_e(m|x)\bigg\} \nonumber \\ 
    &=& \sum_{x|x_y=d_y^*} p_e(m|x) ,
\eea 
where $d^*_y$ is from the set $\{1,\cdots,d\}$. As a consequence,
\bea 
    \chi(m) &=& \sum_{x|x_1=d_1^*} p_e(m|x_1=d_1^*,x_2,\cdots,x_n) \nonumber \\
   && +\sum_{x|x_2=d_2^*} p_e(m|x_1,x_2=d_2^*,\cdots,x_n)+\cdots \nonumber \\
   && + \sum_{x|x_n=d_n^*} p_e(m|x_1,\cdots,x_{n-1},x_n=d_n^*)  .
\eea 
 Each of the $n$ terms on the right-hand side of the above equation includes the probabilities $p_e(m|x_1=d_1^*,x_2=d_2^*,\cdots,x_n=d_n^*)$, as the summation encompasses all $x_i$, except for one $x_y$, which is already fixed to $d_y^*$. 
Thus, in the right-hand side of the above equation, the term $p_e(m|x_1=d_1^*,x_2=d_2^*,\cdots,x_n=d_n^*)$ appears $n$ times and any other term can appear at most $(n-1)$ times. This implies, 
\bea  \label{chim}
\chi(m) &\leqslant &  p_e(m|x_1=d_1^*,x_2=d_2^*,\cdots,x_n=d_n^*) \nonumber \\
&& + (n-1) \sum_{x}p_e(m|x) .
\eea  
Moreover, using the obvious fact
\be
p_e(m|x_1=d_1^*,x_2=d_2^*,\cdots,x_n=d_n^*) \leqslant \max_x p_e(m|x)
\ee
in \eqref{chim}, we obtain Eq. \eqref{bc} for the particular value of $m$. This argument and, thus, Eq. \eqref{bc} holds for any $m$.
Subsequently, summing over $m$ on both sides of \eqref{bc} yields 
    \be\label{ac}
    \sum_{m} \chi(m) \leqslant \sum_{m} \sum_x(n-1) p_e(m|x) +\sum_{m} \max_x p_e(m|x),
    \ee
which simplifies to
\be\label{ec}
    \sum_{m} \chi(m) \leqslant (n-1)d^n+d^n\mathcal{D}_C
\ee
because of \eqref{scd2}. Finally, by replacing the above upper bound on $\sum_m \chi(m)$ into \eqref{cc1}, we get 
\be \label{scnd22}
\mathcal{S}_C(n,d) \leqslant \frac{(n-1)d^n+d^n\mathcal{D}_C}{nd^n} ,
\ee 
which is equivalent to \eqref{scnd}.
\end{proof}


Figure \ref{fig:rac} displays a thorough investigation of quantum advantages for RACs with parameters $(n=2, d=3)$, $(n=3, d=2)$, and $(n=2, d=3)$. Following this, we show that the quantum advantage in RAC can be unbounded.


\begin{figure}[h!]
    \centering
        \includegraphics[width=0.5\textwidth]{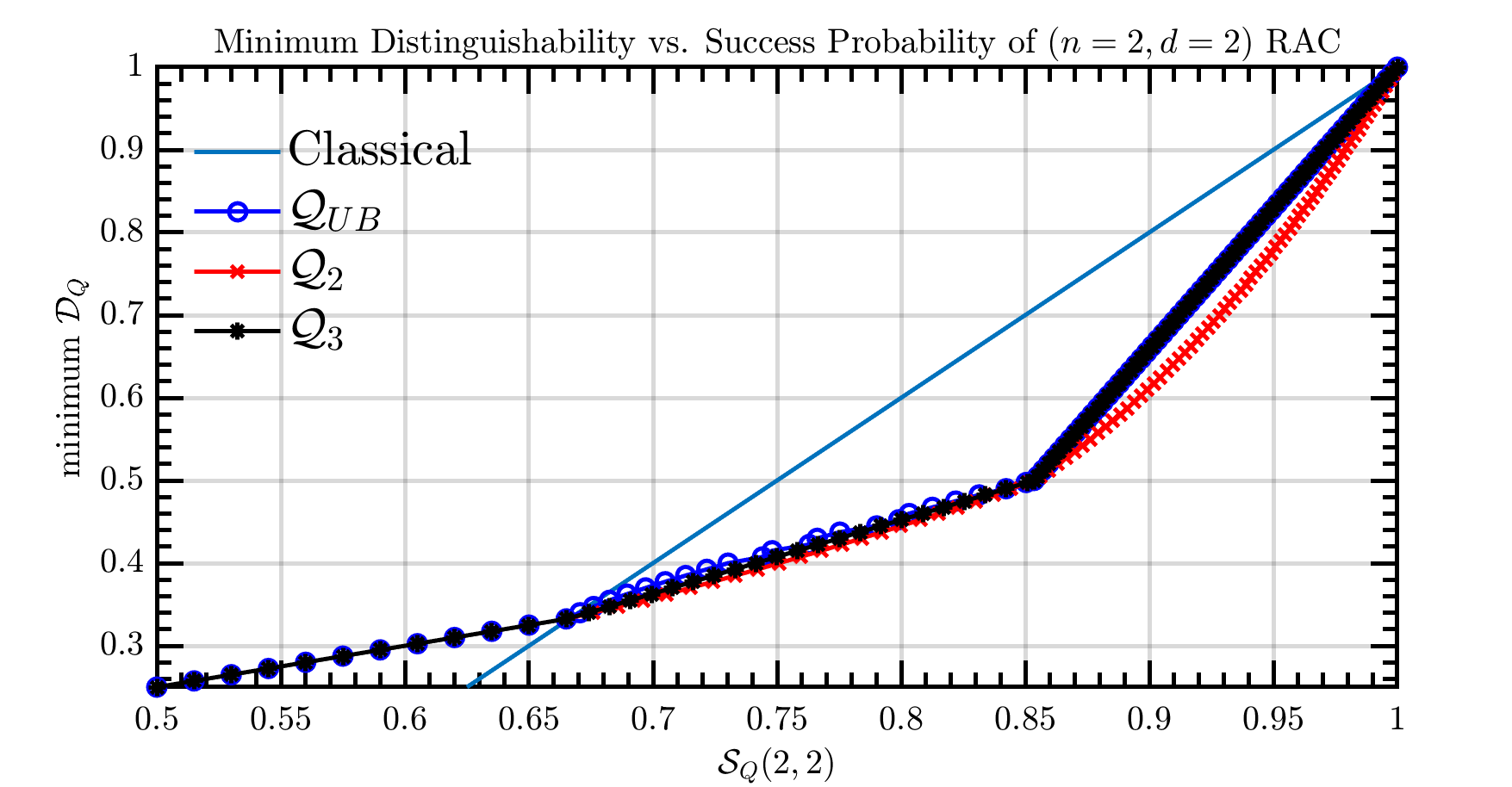}
        \includegraphics[width=0.5\textwidth]{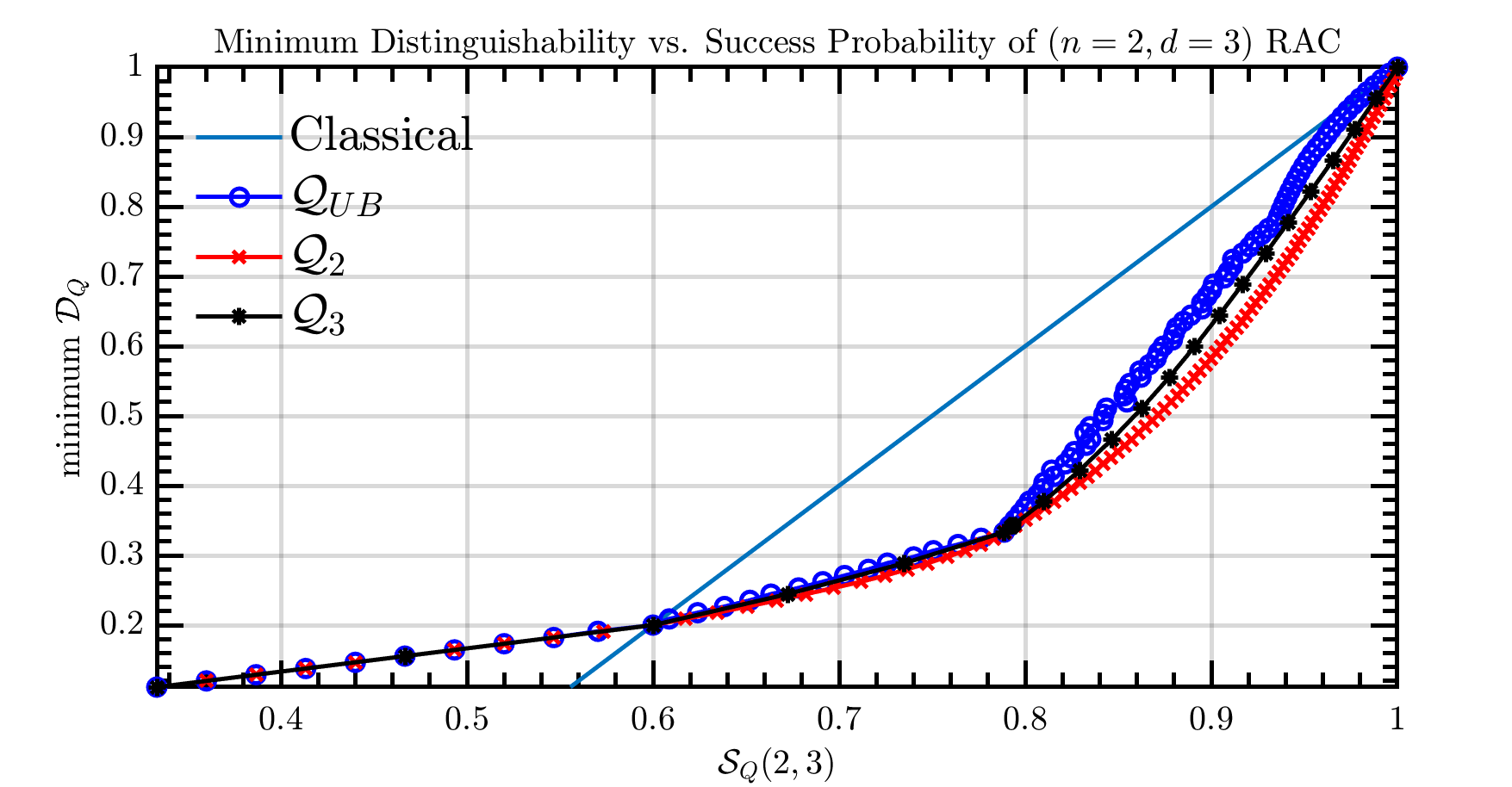}
        \includegraphics[width=0.5\textwidth]{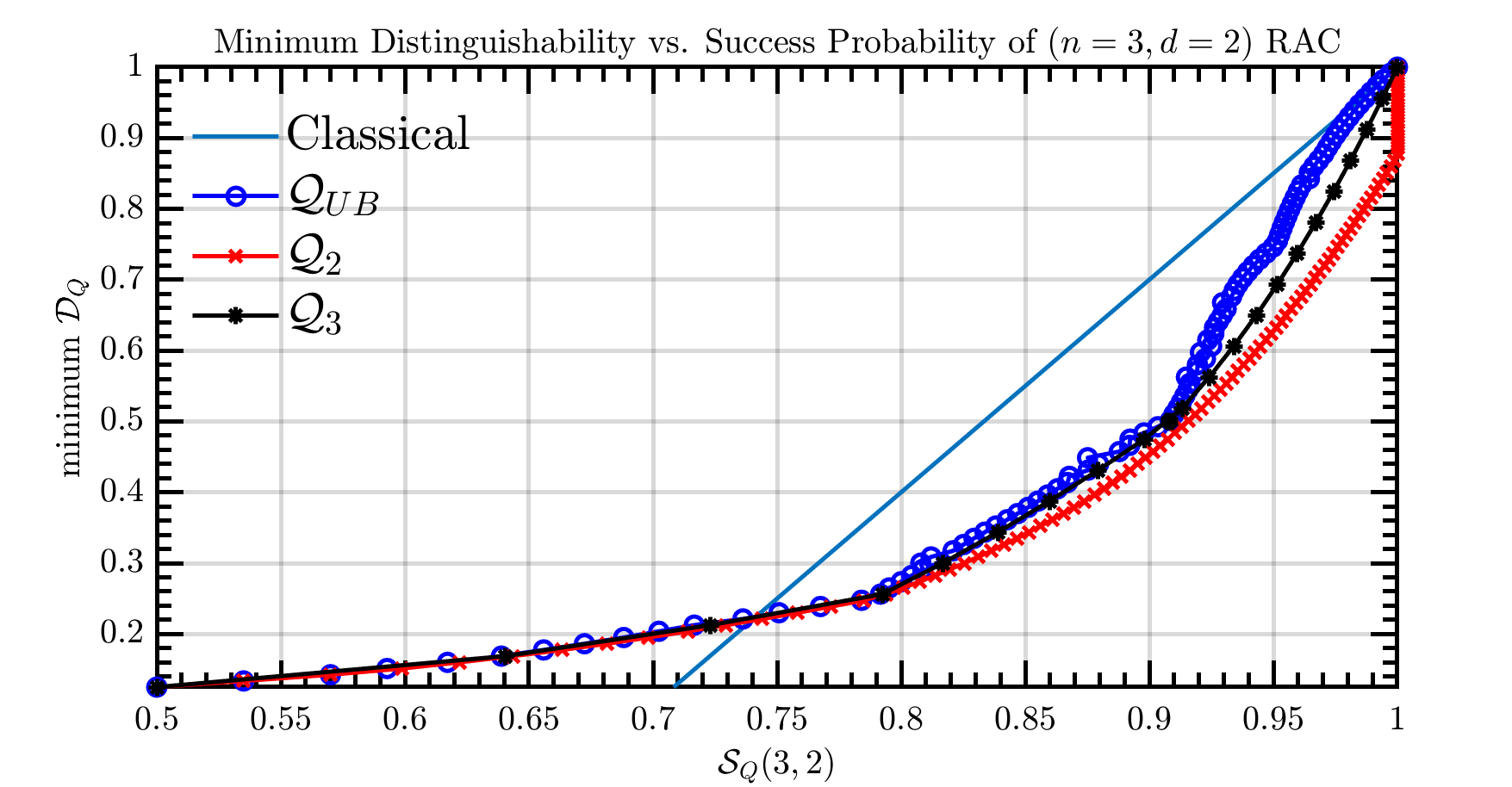}
    
    \caption{Distinguishability in classical and quantum communication vs. the average success probability is presented for three different RACs. The lower bound on $\mathcal{D}_C$ from Eq. \eqref{scnd} is labelled as ''Classical''. By varying values of the success probability $\mathcal{S}(n,d)$, the upper bound on the minimum $\mathcal{D}_Q$ is determined by implementing the see-saw semi-definite programming method presented in Appendix \ref{SDP}. Referred to as $\mathcal{Q}_{UB}$, this bound is obtained for $d^n$-dimensional quantum systems and is certainly achievable in quantum theory. The other two lower bounds are derived through the semi-definite hierarchy presented in Appendix \ref{SDP},  up to the second and third levels, denoted as $\mathcal{Q}_2$ and $\mathcal{Q}_3$. These lower bounds on minimum $\mathcal{D}_Q$ are valid for arbitrary dimensional quantum states. The quantum advantage is present whenever $\mathcal{Q}_{UB}$ falls below the classical lower bound. Notice, that for some values of $\mathcal{S}(n,d)$, the classical values are less than $\mathcal{Q}_{UB}$ in the plots, as Eq. \eqref{scnd} only provides a lower bound on $\mathcal{D}_C$, which is not tight.
    }
    \label{fig:rac}
\end{figure}

\begin{thm}\label{th2}
    There exist quantum strategies for random access codes with $n=2$ such that 
    \be\label{qs1}
\frac{\mathcal{D}^\mathcal{S}_C}{  \mathcal{D}^\mathcal{S}_Q } \geqslant \sqrt{d}, \ \text{ where } \ \mathcal{S}(2,d) = \frac{1}{2} \left( 1 +\frac{1}{\sqrt{d}} \right).
    \ee 
\end{thm}
\begin{proof}
There exists a known quantum strategy involving $d$-dimensional quantum states and measurements for the case of $n=2$ and arbitrary $d$. This strategy achieves $\mathcal{S_Q}(2,d)= 1/2 \left( 1 +1/\sqrt{d} \right)$ by employing measurements performed by Bob in two mutually unbiased bases in $\mathbbm{C}^d$ \cite{rac2d}. According to \eqref{Dd/N}, this quantum strategy must adhere to the constraint $\mathcal{D}_Q\leqslant 1/d$.
Conversely, if we set $n=2$ and $\mathcal{S}_C = 1/2 \left(1 + 1/\sqrt{d}\right)$ in \eqref{scnd}, we deduce that $1/\sqrt{d} \leqslant \mathcal{D}_C$. Consequently, \eqref{qs1} holds.
\end{proof}

\section{Equality problem defined by graphs} \label{IV}

The communication task is defined by an arbitrary graph $G$ having $N$ vertices.  Alice and Bob receive input from the vertex set of $G$, i.e., $x,y\in [N]$, sampled from the uniform distribution. Let $G_x$ denote the set of vertices in $G$ that are adjacent to $x$, with $N_x$ representing the number of vertices adjacent to $x$. In the task, there is a promise that $x=y$ or $x \in G_y$. Bob's aim is to differentiate between these two cases, giving the correct output $z=1$ if $x=y$ and $z=2$ if $x\in G_y$. Hence, the success metric,  
\be \label{SG}
\mathcal{S}(G) = \frac{1}{\sum_x N_x+N} \sum_{y=1}^N \left( p(1|x=y,y) + \sum_{x\in G_y} p(2|x,y) \right) .
\ee 
The quantum advantage for this task in terms of the dimension of the communicated systems has been studied in \cite{saha2019,saha2023}. Concerning the distinguishability of the sender's input, we have the following result.


\begin{thm}
For any graph $G$, the following holds in classical communication,
\be\label{scg}
\frac{1}{N\alpha(G) } \left( \left(\sum_x N_x + N \right)\left(\mathcal{S}_C(G) -1 \right) + N \right) \leqslant \mathcal{D}_C ,
\ee 
where $\alpha(G)$ is the independence number of the graph $G$.
\end{thm}
    \begin{proof} First, we express the average success metric \eqref{SG} in terms of its general form \eqref{SGen} by defining
    \be 
    c(x,y,z)= \frac{1}{(\sum_x N_x + N)} \times 
 \begin{cases}
  1 ,& \text{if $z=1, x=y$}\\
   1 ,& \text{if $z=2, x \in G_y$}\\
    $0$,            & \text{otherwise.} 
\end{cases}
\ee  
Substituting these values of $c(x,y,z)$ into \eqref{scdz2}, we obtain
  \bea \label{scg2}
  \mathcal{S}_C(G) &=& 1-\min_{p_e(m|x)}\frac{1}{\sum_x N_x+N} \times \nonumber \\
  && \sum_m\Bigg[\underbrace{\sum^{N}_{y=1}\min \Bigg\{p_e(m|x=y),\sum_{x\in G_y} p_e(m|x)\Bigg\}}_{\chi(m)}\Bigg]. \nonumber \\
  \eea  
The underlined term in \eqref{scg2} is denoted by $\chi(m)$. The distinguishability of input variable $x$ in this task is given by,
\be\label{scd3}
\mathcal{D}_C=\frac{1}{N}\sum_m \max_{x}\bigg\{ p_e(m|x)\bigg\}.
\ee
For this proof, our primary goal is to show
\be\label{ge}
\chi(m) + \alpha(G)\max_x p_e(m|x) \geqslant \sum_x p_e(m|x), \ \  \forall m.
\ee
Let us fix a value of $m$. For any given encoding $\{p_e(m|x)\}$, consider the set $T \subseteq [N]$ such that 
\be 
\min \Bigg\{p_e(m|y),\sum_{x\in G_y} p_e(m|x)\Bigg\} = \sum_{x\in G_y} p_e(m|x), \ \forall y \in T .
\ee
Subsequently, we can write 
\be \label{chimGsim}
\chi(m) = \sum_{y\notin T} p_e(m|y) + \sum_{y\in T} \left( \sum_{x\in G_y} p_e(m|x) \right) .
\ee 
Further, we divide the set $T$ into two partitions: First, $T_I$ containing the vertices that do not share any common edge, and second, $T\setminus T_I$. 
Note that the subset $T_I$ must be an independent set of the induced subgraph made of all the vertices belonging to $T$. 
For every $y \in T\setminus T_I,$ there exists another $y'  \in T\setminus T_I$ such that $y \in G_{y'}$. This observation allows us to express \eqref{chimGsim} as
\bea 
\chi(m) &  = & \sum_{y\notin T} p_e(m|y) + \sum_{y\in T\setminus T_I} p_e(m|y) \nonumber \\
&& + \sum_{y\in T_I} \left( \sum_{x\in G_y} p_e(m|x) \right) \nonumber \\
& \geqslant & \sum_{y\notin T} p_e(m|y) + \sum_{y\in T\setminus T_I} p_e(m|y) .
\eea  
Adding $\sum_{y\in T_I} p_e(m|y)$ on both sides of the above equation, we get
\be \label{chimb1}
\chi(m) + \sum_{y\in T_I} p_e(m|y) \geqslant \sum_{y=1}^N p_e(m|y) .
\ee 
For any given set $T$, the cardinality of the subset $T_I$ is, at most, $\alpha(G)$, the independence number of $G$. This implies
\be 
\alpha(G) \max_y p_e(m|y) \geqslant \sum_{y\in [T_I]} p_e(m|y) .
\ee 
By substituting the above upper bound in \eqref{chimb1}, we recover \eqref{ge}. Subsequently, summing over $m$ on both sides of \eqref{ge} results 
    \be\label{ac1}
    \sum_{m} \chi(m) + \alpha(G)\sum_{m} \max_x p_e(m|x) \geqslant \sum_{m} \sum_x p_e(m|x),
    \ee
which, because of \eqref{scd3}, reduces to
    \be\label{ec1}
    \sum_{m} \chi(m) \geqslant  N(1-\alpha(G)\mathcal{D}_C ) .
    \ee
Finally, by replacing the above upper bound on $\sum_m \chi(m)$ into \eqref{scg2} we get
\be 
\mathcal{S}_C(G) \leqslant  1- \frac{N(1-\alpha(G)\mathcal{D}_C)}{\sum_x N_x+N} ,
\ee 
which, after some rearrangements, becomes \eqref{scg}.
\end{proof}

An important inquiry is to identify the simplest graph where quantum methods exhibit an advantage over classical ones. Based on the classical constraint given by \eqref{scg}, we observe that any graph with vertices up to four will have no quantum advantage. In order to show there is no quantum advantage, it is sufficient to consider non-isomorphic graphs. Among connected graphs, there are six non-isomorphic graphs with four vertices and two non-isomorphic graphs with three vertices (Fig. \ref{fig:8graphs}).
\begin{figure}[hbt!]
        \centering
        \includegraphics[scale=0.9]{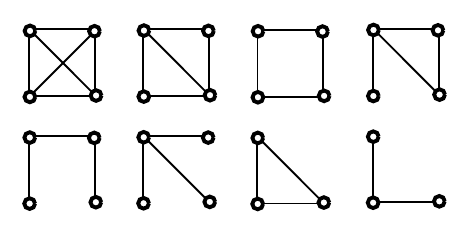}
\caption{Non-isomorphic connected graphs with three and four vertices.}
        \label{fig:8graphs}
\end{figure}

For each of these eight graphs, we use SDP hierarchy presented in Appendix \ref{SDP}, to obtain an upper bound on the value of $(\mathcal{D}_C - \mathcal{D}_Q)$ for all possible values of $\mathcal{S}(G)$. This optimization yields the value of zero, which reveals that there is no advantage for any of these graphs, irrespective of the values of $\mathcal{S}(G)$. We will now show that the communication task defined by a graph with five vertices does exhibit quantum advantages.

Before delving into this, let us consider a form of quantum strategy for the equality problem given a graph $G$. The strategy is defined by a set of $N$ quantum states $\{\ket{\psi_x}\}_{x=1}^N \in \mathbbm{C}^d,$ with 
\be \label{qspsi}
    \rho_x = |\psi_x\ra\!\la \psi_x|, \ M_{1|y} = |\psi_y\ra\!\la \psi_y|, \ M_{2|y} = \I - |\psi_y\ra\!\la \psi_y|  ,
\ee 
where $\rho _x$ is the quantum state sent by Alice for input $x$, and the binary outcome measurements performed by Bob on the received quantum state for input $y$ are expressed as $\{M_{1|y},M_{2|y}\}$. For this strategy, a straightforward calculation leads to 
\be \label{sqn}
    \mathcal{S}_Q (G) =     
1 - \frac{1}{\sum_x N_x+ N} \sum_y\sum_{x\in G_y}|\bra{\psi_x}\psi_y\rangle|^2 .
\ee
Simultaneously, it holds that $\mathcal{D}_Q \leqslant d/N$.

\begin{thm} \label{th4}
Let us consider $N$-cyclic graph (denoted by $\Delta_N$) such that $N\geqslant 5$ and odd. There exist quantum strategies with two-dimensional systems such that 
    \be \label{ratio}
\frac{\mathcal{D}^\mathcal{S}_C(\Delta_N)}{\mathcal{D}^\mathcal{S}_Q(\Delta_N)} = \frac{N}{N-1}\left({1-2\sin^2\left(\frac{\pi}{2N}\right)}\right) > 1,
    \ee 
    for $\mathcal{S} = 1-\frac{2}{3}\sin^2\big(\frac{\pi}{2N}\big)$.
\end{thm}
\begin{proof}
Let us consider a quantum strategy given by \eqref{qspsi} where the states $\{\ket{\psi_x}\}_x \in \mathbbm{C}^2$ are given by 
\begin{equation} \label{psixNpd}
    \ket{\psi_x} = \cos\left(\frac{(x-1)\beta}{2}\right)\ket{0}+
               \sin\left(\frac{(x-1)\beta}{2}\right)\ket{1} ,
\end{equation}
and $\beta=\big(\frac{(N-1)\pi}{N}\big)$.
These states form a regular $N$-gon in the $X$-$Z$ plane of the Bloch sphere such that the angle between the Bloch vectors of $\ket{\psi_x}$ and $\ket{\psi_{x\pm 1}}$ is $\beta$. The Bloch vectors of these states for $N=5$ and $N=7$ are shown in Fig. \ref{fig:states}.
\begin{figure}[hbt!]
    \centering
    \includegraphics[scale=0.31]{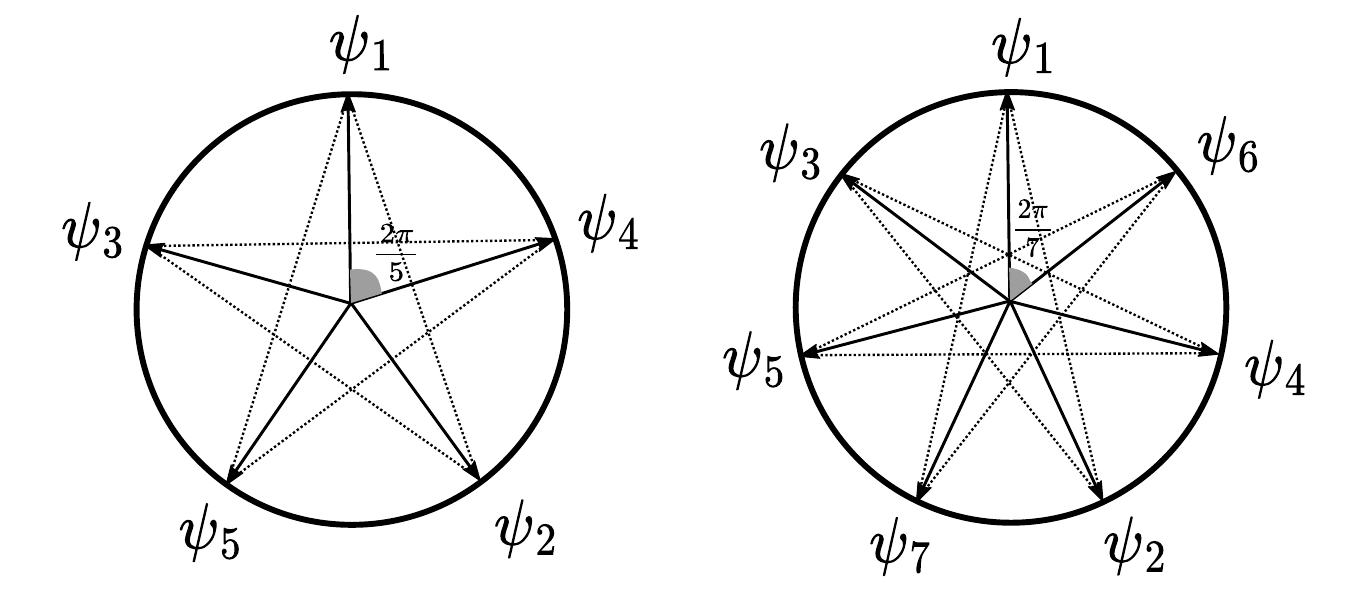}
    \caption{Bloch vectors of qubit states given by \eqref{psixNpd} for $N=5,7$.}
    \label{fig:states}
\end{figure}

For this quantum strategy, $|\bra{\psi_x}\psi_y\rangle|=\sin\big(\frac{\pi}{2N}\big)$ for any pair of $x,y,$ that are connected by an edge. By substituting the expression in \eqref{sqn} and $N_x=2, \forall x$, we obtain
\be\label{sqgraph}
\mathcal{S}_Q (\Delta_N) = 1-\frac{2}{3}\sin^2\Big(\frac{\pi}{2N}\Big).
\ee
Replacing this value of $\mathcal{S}_Q$ in \eqref{scg} with $\alpha(\Delta_N)=(N-1)/2$ and other specifications of $N$-cyclic graph, we get
\be\label{scdn1}
\mathcal{D}^\mathcal{S}_C(\Delta_N) \geqslant \frac{2}{N-1}\Big(1-2\sin^2\frac{\pi}{2N}\Big).
\ee
From the fact that 
$\mathcal{D}^\mathcal{S}_Q(\Delta_N) \leqslant 2/N$ and the above relation \eqref{scdn1}, we obtain \eqref{ratio}. Figure \ref{fig:plot} depicts how the ratio \eqref{ratio} changes with $N$.
\end{proof}  
\begin{figure}[hbt!]
    \centering
\includegraphics[scale=0.35]{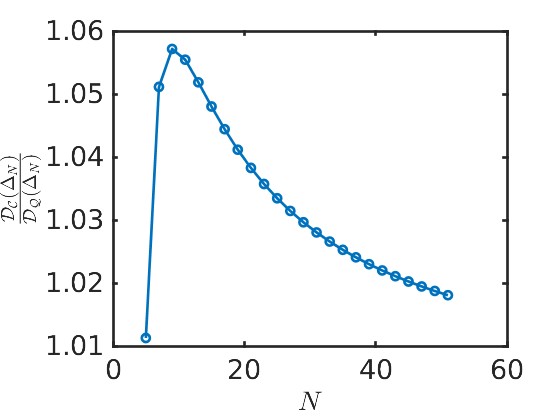}
    \caption{The ratio between classical and quantum distinguishability of input states as a function of $N$ for the communication tasks defined by regular $N$-cyclic graph, where $\mathcal{S}=1-(2/3)\sin^2(\pi/2N)$ with $N\geqslant5$ and odd. The ratio reduces to $1$ as $N\rightarrow{\infty}$.}
    \label{fig:plot}
\end{figure}

Consider the set of vectors in $\mathbbm{C}^d$ of the form 
\be \label{xs}
\ket{\psi_x} = (1/\sqrt{d})\left[(-1)^{x_1},(-1)^{x_2}, \ldots, (-1)^{x_d} \right]^T,
\ee 
where every $x_i\in \{0,1\}$.
This set comprises $2^d$ distinct vectors, and two vectors are orthogonal when the values of $x_i$ differ in exactly $d/2$ places. The orthogonality relations among these vectors can be represented through a graph, commonly known as the Hadamard graph $(H_d)$, where the vertices symbolize the vectors, and two vertices are connected by an edge if the respective vectors are orthogonal. Consequently, one can consider the communication task based on $H_d$, previously introduced as the distributed Deutsch-Jozsa task \cite{RevModPhys.82.665}. Hereafter, we demonstrate that the ratio $\mathcal{D}^\mathcal{S}_C/\mathcal{D}^\mathcal{S}_Q$ experiences exponential growth with $d$.

\begin{thm}
    For the Hadamard graph $H_d$ where $d$ is divisible by 4,
\be \label{Hd}
\frac{\mathcal{D}^\mathcal{S}_C}{  \mathcal{D}^\mathcal{S}_Q } \geqslant \frac{(1.005)^d}{d}  , \text{ where } \mathcal{S} = 1 .
\ee 
\end{thm}
\begin{proof}
Clearly, the quantum strategy of the form \eqref{qspsi} such that $\ket{\psi_x}$ given by \eqref{xs}, perfectly accomplishes the task, leading to $\mathcal{S}(H_d) = 1$. Additionally, in this strategy, as per \eqref{fact:d/N}, $\mathcal{D}_Q \leqslant d/2^d$. On the classical side, according to \eqref{scg}, we know that $\mathcal{S}_C(H_d)=1$ only if $\mathcal{D}_C=1/\alpha(H_d)$. The result established by Frankl-R\"{o}dl \cite{FR} implies that $\alpha(H_d) \leqslant 1.99^d$, specifically when $d$ is divisible by $4$ (refer to Theorem~1.11 in \cite{FR}). As a consequence, we have \eqref{Hd}.
\end{proof}

It is worth observing that the expression on the right-hand side of \eqref{Hd} is not an increasing function for smaller values of $d$, but it exponentially increases for $d\geqslant 1124$. In particular, $d\geqslant 1124 < 2^{11}$, implies that we require less than $11$-qubits which have been physically realized on quantum computers \cite{Benchmarking,PhysRevX.12.031010,PhysRevX.8.021012,16-qubit}. In quantum communication, genuine $10$-qubit quantum states have been experimentally realized and semi-device independently certified \cite{PhysRevLett.102.010401}. To cite a realisation of the unbounded advantage, we calculate the ratio of advantage for $15$-qubit system and find out $(\mathcal{D}^\mathcal{S}_C/\mathcal{D}^\mathcal{S}_Q ) \approx 10^{66}$. Such an quantum advantage is near-term realizable.

Moving forward, we will outline sufficient criteria for a graph to exhibit quantum advantage in the task defined by that graph.

\begin{thm}
    For any graph $G$, say $\beta(G)$ is the minimum dimension $d$ such that there exists $N$ number of quantum states $\{\ket{\psi_x}\}_{x=1}^N\in \mathbbm{C}^{d}$ satisfying the orthogonality relations according to $G$, that is, $\la \psi_x|\psi_y\ra = 0$ for every pair of vertices $x,y,$ that are connected in $G$. Then,
\be \label{GS1}
\frac{\mathcal{D}^\mathcal{S}_C}{  \mathcal{D}^\mathcal{S}_Q } \geqslant \frac{N}{\alpha(G)\beta(G)}  , \ \  \text{ for } \mathcal{S} = 1 .
\ee  
\end{thm}
\begin{proof}
    If there exist states $\{\ket{\psi_x}\}_x\in \mathbbm{C}^{d}$ that fulfill the orthogonality conditions according to $G$, then the quantum strategy given by \eqref{qspsi} achieves $\mathcal{S}(G)=1$, with $\mathcal{D}_Q \leqslant d/N$. Conversely, according to \eqref{scg}, for $\mathcal{S}_C(G)=1$, it follows that $\mathcal{D}_C \geqslant 1/\alpha(G)$. Combining these two observations yields \eqref{GS1}.
\end{proof}
The aforementioned result suggests that a graph offers quantum advantages for achieving the respective tasks perfectly ($\mathcal{S}=1$) if 
\be \label{GS1c}
\frac{N}{\alpha(G)\beta(G)}>1 .
\ee 
Interestingly, a graph exhibits state-independent contextuality if and only if its fractional chromatic number, denoted as $\chi_f(G)$, is greater than $\beta(G)$, that is, $\chi_f(G) > \beta(G)$ \cite{Ramanathan2014}. Additionally, for any graph, $\chi_f(G) \geqslant N/\alpha(G)$. Consequently, any graph satisfies the condition \eqref{GS1c} meets the state-independent contextuality criterion. It is important to note that the converse is not universally true. However, for vertex-transitive graphs, the reverse implication does hold, as these graphs adhere to the equality $\chi_f(G) = N/\alpha(G)$. Notably, the smallest graph meeting the condition \eqref{GS1c} was identified as the Yu-Oh graph, comprising 13 vertices \cite{PhysRevLett.114.250402}.

\section{Pair distinguishability task} \label{V}

We consider a generalized version of the task introduced in \cite{bod}. Alice receives input $x \in [N]$ with $p_x=1/N$ and Bob receives a pair of inputs $y \equiv (x,x')$ randomly such that $x,x' \in \{ 1, \cdots, N \}$ and $x< x'$. The task is to guess $x$. In other words, given Bob's input $(x,x')$, the task is to distinguish between these two inputs. Subsequently, the average success metric, 
\be
\mathcal{S}(N) = \frac{1}{N(N-1)} \sum_{\substack{x,x' \\  x<x' }} \big[ p(x|x,(x,x')) + p(x'|x',(x,x')) \big] .
\ee  

\begin{thm} The following holds for classical communication, 
\be \label{scpd}
(N-1)(\mathcal{S}_C(N) -1) +1 \leqslant  \mathcal{D}_C .
\ee 
\end{thm}
\begin{proof} 
For this task, the form of success metric \eqref{scdz2} reduces to 
\bea \label{scpdg}
\mathcal{S}_C(N) &=& 1 - \min_{\{p_e(m|x)\}} \frac{1}{N(N-1)} \times \nonumber \\
&&   \sum_m \sum_{\substack{x,x' \\  x<x' }} \min \left\{ p_e(m|x), p_e(m|x') \right\} ,  
\eea 
after substituting the following expression
\be 
c(x,(x,x'),z) = \frac{1}{N(N-1)} \times \begin{cases}   1 , & \text{if $z=x,$ $x\in \{x,x'\}$}\\    $0$,           & \text{otherwise.} \end{cases}
\ee
We use the fact that any set of non-negative numbers $\{a_i\}$ with $i=1,\cdots,N$, satisfies the relation
\be  \label{relation1}
\sum_{i<j} \min\{a_i,a_{j}\}  \geqslant \sum_i a_i - \max\{a_1,\cdots,a_N\} .
\ee 
Replacing $a_x$ by $p_e(m|x)$ into the above relation, we get
\bea   
&& \sum_{x<x'} \min\{p_e(m|x), p_e(m|x') \}  \nonumber \\
& \geqslant & \sum_x p_e(m|x) - \max_x\{p_e(m|x) \}.
\eea  
Next, we take summation over $m$ on both sides and replace the expression of distinguishability \eqref{DC} to obtain
\be \label{pdt1}
\sum_m \sum_{x<x'} \min\{p_e(m|x), p_e(m|x') \} \geqslant N (1- \mathcal{D}_C) .
\ee 
The left-hand side of the above inequality appears on the right-hand side of \eqref{scpdg}. Hence, by substituting the lower bound from \eqref{pdt1} into \eqref{scpdg}, we deduce \eqref{scpd}.
\end{proof}
Let us discuss a quantum strategy for this task that provide advantages. One interesting feature of this task is the fact that $\mathcal{S}_Q$ is fixed by the set of quantum state $\{\rho_x\}$ communicated by Alice. Specifically, due to the Helstrom norm \cite{Helstrom},
\be \label{sqpdt}
\mathcal{S}_Q = \frac12 + \frac{1}{2N(N-1)} \sum_{\substack{x,x' \\  x<x' }}  \| \rho_x - \rho_{x'} \|.
\ee 
%
%
Consider the qubit states $\ket{\psi_x}=\cos{\left(\frac{x\pi}{N}\right)}\ket{0} +\sin{\left(\frac{x\pi}{N}\right)}\ket{1}$, where $x \in [N]$. We evaluate $\mathcal{S}$ from \eqref{sqpdt} for these states and further obtain the corresponding $\mathcal{D}_C$ from \eqref{scpd}, which are given as follows:
\begin{center}
\begin{tabular}{|c c c c|} 
 \hline
 $N$ & $\mathcal{S}$ & $\mathcal{D}_C$ & $~ ~ \mathcal{D}_Q$  \\ [0.6ex] 
 \hline\hline
 3 & 0.933 & $~ ~ \geqslant 0.866$ & $~ ~  \frac23 ~$ \\ [0.6ex]
 \hline
 4 & 0.9 & $~ ~ \geqslant 0.7$ & $~ ~ \frac12 ~$ \\ [0.6ex]
 \hline
 5 & 0.8847 & $~ ~ \geqslant 0.5388$  & $~ ~ \frac25 ~$\\ [0.6ex]
 \hline
 6 & 0.8732 & $~ ~ \geqslant 0.366$  & $~ ~ \frac13 ~$\\ [0.6ex]
 \hline
\end{tabular}
\end{center}
We produce a more detailed investigation as Fig. \ref{fig:plot2} presents the quantum advantages over classical communication in this particular task for $N=3$ and $N=4$ with qubit states.
\begin{figure}[h!]
    \centering
\includegraphics[width=0.45\textwidth]{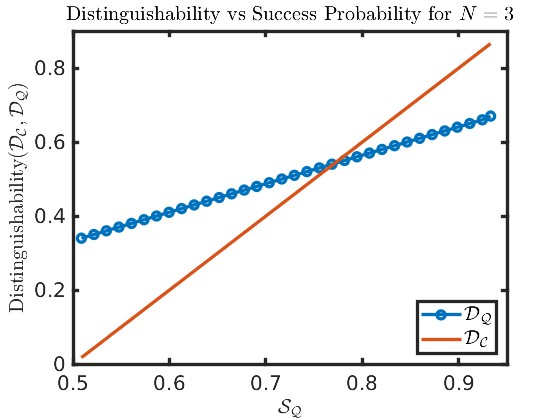}
\includegraphics[width=0.45\textwidth]{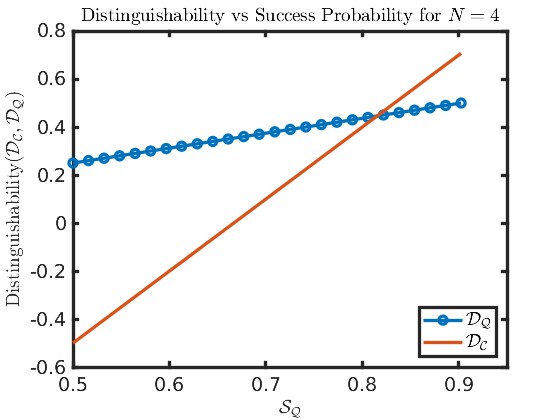}
    \caption{Distinguishability in classical and quantum communication vs average success probability of the pair distinguishability task is presented here for $N=3$ and $N=4$ with qubit states. For different values of $\mathcal{S_Q}$, taken in small intervals, the values of $\mathcal{D_Q}$ are obtained along with the quantum states by implementing the see-saw method of semi-definite optimization presented in Appendix \ref{SDP}. From \eqref{scpd}, we find the lower bound on $\mathcal{D_C}$ to get the same average success probability. We observe that if we increase the dimension of the states, the values of $\mathcal{D_Q}$ remain the same.}
    \label{fig:plot2}
\end{figure}

\section{Comparison with the communication complexity measured by minimum dimension}

In this section, we compare our measure of communication complexity, minimum distinguishability, to the more commonly used measure of communication complexity, minimum dimension. Let us denote by $d_C$ the dimension of the message in classical communication and $d_Q$ for the dimension of quantum message in quantum communication. Let $\overline{\mathcal{S}}_{d_C}$ denote the maximum classical value of the success metric \eqref{SGen} when $d_C$ dimensional classical messages are used in communication. Let us first make the following general observation, upper-bounding $\overline{\mathcal{S}}_{d_C}$ with the maximum value of the success metric $\mathcal{S}_C$ with constrained distinguishabiltiy,

\begin{obs}\label{obs2}
    For any communication task, the maximum value of the success metric with $d_C$ dimensional classical communication, $\overline{\mathcal{S}}_{d_C}$, is upper-bounded by the maximum value of the success metric under distinguishability constraint, $\mathcal{S}_C$, where $\mathcal{D} \leqslant d_C/N$, i.e., for $\mathcal{D} \leqslant d_C/N$,
   \be \label{obs2eq}
        \overline{\mathcal{S}}_{d_C}\leqslant \mathcal{S}_C
   \ee
\end{obs}
\begin{proof}
    The observation follows directly from the fact that any $d_C$ dimensional classical strategy the distinguishability, $\mathcal{D} \leqslant d_C/N$. However, there might exist higher dimensional $(> d_C)$ classical strategies that satisfy this distinguishability constraint. 
\end{proof}
Hence, if we have a $d_Q=d_C$-dimensional quantum strategy with $\mathcal{D}_Q= \frac{d_C}{N}$ which provides an advantage in terms of distinguishability, i.e., $\mathcal{S}_Q>\mathcal{S}_C$, where $\mathcal{S}_C$ is the maximum classical value of the success metric under distinguishability constraint $\mathcal{D} \leqslant d_Q/N$, then \eqref{obs2eq} implies it also provides an advantage in terms of dimension, i.e., $\mathcal{S}_Q>\overline{\mathcal{S}}_{d_C}$. Consequently, to find a strategy for which we have an advantage in terms of distinguishability but not in terms of dimension, we have to consider quantum strategies with $\mathcal{D}_Q<\frac{d_C}{N}$, which leads us to our next observation: 
\begin{obs}
    There exists quantum strategies that offer advantage in a communication complexity task in terms of distinguishability but not in terms of dimension.
\end{obs}
\begin{proof}
For the RACs considered in Fig. \ref{fig:rac}, the best success probabilities using two-dimensional classical messages are 3/4, 3/4 and 2/3, respectively, for the three plots. On the other hand, we find quantum advantages (Fig. \ref{fig:rac}) even when the success probabilities $\mathcal{S}$ are less than these values. Therefore, our study of RACs reveals instances where quantum advantage in dimension is absent, but the quantum advantage in distinguishability exists. For example, for $n=2, d=2$, Alice encodes her inputs into four states such that,
     \be
     \rho_{x_1x_2} = p\ket{\psi_{x_1x_2}}\bra{\psi_{x_1x_2}}+(1-p)\frac{\I}{2},
     \ee
     where, $\ket{\psi_{00}}=\cos\frac{\pi}{8}\ket{0}+\sin\frac{\pi}{8}\ket{1}, \ket{\psi_{01}}=\cos\frac{3\pi}{8}\ket{0}+\sin\frac{3\pi}{8}\ket{1}, \ket{\psi_{10}}=\cos\frac{7\pi}{8}\ket{0}+\sin\frac{7\pi}{8}\ket{1}$ and $\ket{\psi_{11}}=\cos\frac{5\pi}{8}\ket{0}+\sin\frac{5\pi}{8}\ket{1}.$ To redeem the first bit, Bob  measures in $\{\ket{0},\ket{1}\}$ basis and for the second bit, he measures in $\{\ket{+},\ket{-}\}$ basis. In this strategy with $p=4/7$, $\mathcal{S}_Q=0.7$ such that $\mathcal{D}_Q=0.3859$. On the other hand, from \eqref{scnd22}, we have $\mathcal{S}_C\leqslant 0.693$ with $\mathcal{D}_C=0.3859$. But, we know $\overline{\mathcal{S}}_{d_C=2}=0.75$. So, there is advantage in quantum communication in terms of minimum distinguishability but not in terms of minimum dimension.  
\end{proof}


The pair distinguishability tasks in terms of dimension have been previously discussed \cite{PhysRevLett.110.150501}. In this case, for $d=2$,
\be \label{brunner1}
\overline{\mathcal{S}}_{d_C}=\frac12-\frac{1}{N(N-1)}\left(\left\lfloor\frac{N}{2}\right\rfloor\left(N-\left(\left\lfloor\frac{N}{2}\right\rfloor+1\right)\right)\right),
\ee
where $\lfloor x \rfloor$ denotes the integer part of $x$, and 
${S}_Q=\frac{N}{4(N-1)}$.
From these two relations, it is easy to see that only odd $N$ gives the quantum advantage at $d=2$. for our quantum strategy, we found advantages for $N=3,4,5,6$ in terms of distinguishability. For $N=4,6$, the quantum strategy under the distinguishability constraint yields quantum advantage, but no advantage is obtained under the dimension constraint. 

Let us now consider the equality problem defined by regular polygons with odd number of vertices, the maximum success metric for $d$-levelled classical communication turns out to be $\overline{\mathcal{S}}_{d_C}=1-\frac{2}{3N}$ \cite{Saha_2019}. Due to \eqref{th4}, we can see for a regular polygon with odd $N\geqslant5$, $\frac{\mathcal{S}_Q}{\overline{\mathcal{S}}_{d_C}}=\frac{1-\frac{2}{3}\sin^2\big(\frac{\pi}{2N}\big)}{1-\frac{2}{3N}}\geqslant 1$ with dimension $2$. Thus, the quantum advantage manifests in both dimension and distinguishability.
One can be curious about finding a quantum strategy that gives an advantage with respect to a dimension but not with respect to distinguishability. For example, let us consider the equality problem for the regular pentagon. In this case, $\overline{\mathcal{S}}_{d_C=2}=13/15$ and $\mathcal{S}_C \leqslant 14/15$ with $\mathcal{D}_C=2/5$. As one can observe a gap between these two success metrics, there must be a valid quantum strategy such that the success metric lies between the values of $\frac{13}{15}$ and $\frac{14}{15}$. But for that matter, the bound of $\mathcal{S}_C$ must be proved to be tight. For the time being, we leave this problem for future exploration. \\

Interestingly, there are instances where quantum advantages measured by both dimension and distinguishability follow similar trends as the complexity of the problems increases. For RAC with $n=2$, classical success metric follows the relation $\overline{\mathcal{S}}_{d_C}\leqslant \frac12(1+\frac{d_C}{d^2})$ \cite{PhysRevA.107.062210}. To achieve the success metric of the quantum strategy at \eqref{th2}, $d_C$ must obey the constraint $d_C \geqslant d\sqrt{d}$. So, we can deduce the advantage parameter $\frac{d_C}{d_Q}\geqslant \sqrt{d}$, when $\mathcal{S}=\frac12(1+\frac{1}{\sqrt{d}})$, which is like our previous version of the game defined by distinguishability.
Similarly, for the equality problem defined by the Hadamard graph, it is shown that $\frac{d_C}{d_Q}\geqslant \frac{(1.005)^d}{d}$ to achieve $\mathcal{S}=1$ \cite{RevModPhys.82.665,saha2023}.  


\section{Conclusions}\label{VI}

Communication complexity plays a crucial role in information science, and quantum theory offers a notable advantage over classical methods. Traditionally, one-way communication complexity is quantified by the minimal dimension of systems that the sender uses to achieve a given task. However, in this investigation, we take a fresh approach by evaluating communication complexity based on the distinguishability of the sender's input, without imposing constraints on the dimension of the communicated systems. This perspective gains significance when preserving the confidentiality of the sender's input is paramount. Moreover, the dimension independence nature of this measure implies that quantum advantage signifies something unattainable in classical communication and does not rely on specific details of the physical system.

We concentrate on two significant categories of communication complexity tasks: the general version of random access codes and equality problems defined by graphs. Lower bounds on the distinguishability of the sender's input are derived as a function of the success metric for these tasks in classical communication. Remarkably, we demonstrate exponential and polynomial increases in the ratio between classical and quantum distinguishability, showcasing the unbounded quantum advantage in preserving the sender's input and paving the way for new quantum supremacy in distributed computation.  In Sec. VI, we give a quantitative analysis of the superiority of the communication complexity measured by distinguishability against dimension. There can be an interesting venture to find the class of communication tasks and strategies where the minimum dimension scenario has an edge over the minimum distinguishability scenario.

Other future works could explore advantageous quantum protocols for random access codes with higher $n$ and for equality problems based on different graphs. From our study, it can be observed that the range of the success metric values where quantum advantage occurs depends on the specifics of the communication task. Exploring this range for various tasks would be an intriguing direction for future work. Furthermore, to get tight bounds on classical communication for all values of the success metric additional distiguishability-like constraints such as anti-distinguishability can be considered \cite{chaturvedi2021}.   Additionally, proposing privacy-preserving computational schemes based on these results and exploring other communication complexity tasks with practical applications are avenues for further research. 


\subsection*{Acknowledgements}
AC acknowledges financial support by NCN Grant SONATINA 6 (Contract No. UMO-2022/44/C/ST2/00081). DS acknowledges support from STARS (STARS/STARS-2/2023-0809), Govt.
of India.

\bibliography{ref}

\appendix

\section{Semi-definite programming methods}\label{SDP}
Our aim is to retrieve the maximum quantum value of a generic success metric \eqref{SGen} of a one-way communication task \eqref{SGenQ}, 
\be 
\mathcal{S}_Q = \sum_{x,y,z} c(x,y,z) \tr (\rho_x M_{z|y}) ,
\ee

given an upper bound on the distinguishability of the sender's states, say $\mathcal{D}_Q\leq p$, which translates to the following optimization problem,
\be \label{SvsDQ}
\begin{split}
 \mathcal{S}^{max}_{Q}=\underset{\{\rho_x\}^N_{x=1},\{M_{z|y}\}}{\ \ \max\ \ } & \sum_{x,y,z} c(x,y,z) \tr (\rho_x M_{z|y})\\
\text{s.t.\ \ } &  \mathcal{D}_Q \leqslant p\\
&  \rho_x \geqslant 0, \Tr{\rho_x} = 1, \ \ \forall x \in [N], \\
&  M_{z|y} \geqslant 0, \forall z \in [D], y\in[M], \\
&  \sum_{z}M_{z|y} = \I, \forall y\in[M],
\end{split}
\ee
Observe that the optimal solution of \eqref{SvsDQ}, $\mathcal{S}^{max}_{Q}$ given $\mathcal{D}_Q \leqslant p$, also yields an upper-bound on the minimal distinguishability, namely $\mathcal{D}_Q\geq p$, required to achieve the success metric $\mathcal{S}_{Q}=\mathcal{S}^{max}_{Q}$.
As the computing the distinguishability $\mathcal{D}_Q$ forms a separate SDP, the optimization problem \eqref{SvsDQ} does not yield to standard SDP solutions primarily because of the constraint $\mathcal{D}_Q \leqslant p$. However, Ref. \cite{Tavakoli2022informationally} presented an ingenious technique to enable SDP methods based on the observation of SDP dual of the distinguishability optimization problem \eqref{DistinguishabilityQuantum}, which we describe now. 

Let us define an auxiliary operator $\Theta$ with the property,
\be \label{AuxVar}
   \Theta \geq \rho_x, \ \ \forall x \in [N]. 
\ee
This variable allows us to upper bound the distinguishability \eqref{DistinguishabilityQuantum} in the following way,
\bea \label{AuxVarCon} \nonumber
&\mathcal{D}_Q & = \max_{\{M_x\}} \frac1N \sum_x \tr(\rho_x M_{x}) \\ 
& & \leqslant \max_{\{M_x\}} \frac1N  \tr(\Theta\sum_xM_x)=\frac1N\tr(\Theta),
\eea
where we have used \eqref{AuxVar} for the inequality, and $\sum_x M_{x}= \I$ for the second. Equation \eqref{AuxVarCon} allows us to impose the constraint $\mathcal{D}_Q \leqslant p$ in the optimization problem \eqref{SvsDQ} as a tracial constraint, $\frac1N\tr(\Theta)\leqslant p$, such that the optimization problem \eqref{SvsDQ} now becomes,
\be \label{SvsDQSeeSaw}
\begin{split}
 \mathcal{S}^{max}_{Q_{LB}}=\underset{\{\rho_x\}^N_{x=1},\Theta,\{M_{z|y}\}}{\ \ \max\ \ } & \sum_{x,y,z} c(x,y,z) \tr (\rho_x M_{z|y})\\
\text{s.t.\ \ } &  \rho_x \geqslant 0, \Tr{\rho_x} = 1, \ \ \forall x \in [N], \\
&  \Theta \geq \rho_x, \ \ \forall x \in [N], \\
&  \frac1N\tr(\Theta)\leqslant p,\\
&  M_{z|y} \geqslant 0, \forall z \in [D], y\in[M], \\
&  \sum_{z}M_{z|y} = \I, \forall y\in[M].
\end{split}
\ee
Now, for fixed measurements $\{M_{z|y}\}$ the optimization problem \eqref{SvsDQSeeSaw} becomes a straightforward SDP for the states $\{\rho_{x}\}$ and the auxiliary operator $\Theta$. Similarly, for fixed $\{\rho_{x}\}$ and $\Theta$ the optimization problem is a SDP for the measurements. Thus, for a given Hilbert space dimension, we can alternate these two SDPs, keeping the optimal solution of one as fixed parameters for the other, \'{a} la see-saw, to retrieve dimension dependent lower bounds on the maximum success metric, $S^{max}_{Q_{LB}}\leqslant S^{max}_{Q}$. We also retrieve an upper bound on the minimal distinguishability, $\mathcal{D}_Q\geqslant p$, required to achieve the success metric $\mathcal{S}_{Q}=\mathcal{S}^{max}_{Q_{LB}}$.

However, the optimization for absolute quantum maximal value of the success metric with restricted distinguishability is substantially more arduous to solve, primarily because the dimension of the quantum systems could in principle be arbitrarily large. To facilitate series of tightening dimension independent upper bounds on the absolute maximal value of the success metric, given $\mathcal{D}_Q\leqslant p$, we, now, present a hierarchy of SDP relaxations of the optimization problem \eqref{SvsDQSeeSaw}.

In particular, such an hierarchy was initially formulated in \cite{Tavakoli2022informationally}. Here we present a modified formulation which is less complex and easier-to-implement. The central concept underpinning our SDP relaxations is that of \emph{hinged moment matrices} \cite{Chaturvedi2021characterising}. For any positive semi-definite operator $\tau\geqslant 0$, and any sequence of linear operators $\mathcal{O}=\{O_i\}^N_{i=1}$, the moment matrix $\Gamma^{\mathcal{O}}_\tau$ is defined as $(\Gamma^{\mathcal{O}}_{\tau})_{i,j}\equiv(\Gamma^{\mathcal{O}}_{\tau})_{O_i,O_j}=\Tr(\tau O^\dagger_i O_j)$. Because $\tau$ is positive semi-definite, a moment matrix \emph{hinged} on $\tau$ must be positive semi-definite as well, $\Gamma^{\mathcal{O}}_{\tau} \geqslant 0$.
Now following the prescription in \cite{Chaturvedi2021characterising}, for any given list of operators $\mathcal{O}$, we consider $N$ moment matrices $\{\tau^{\mathcal{O}}_{\rho_x}\}^N_{x=1}$ each hinged on the corresponding $N$ density operators $\{\rho_x\}^N_{x=1}$ representing the sender's $N$ preparations, such that $\tau^{\mathcal{O}}_{\rho_x}\geqslant 0$ for all $x\in[N]$. We consider an additional moment matrix $\tau^{\mathcal{O}}_{\Theta}$ hinged on the auxiliary variable $\Theta$, such that the definition \eqref{AuxVar} implies the following constraint,
\be
    \Gamma^{\mathcal{O}}_\Theta \geq \Gamma^{\mathcal{O}}_{\rho_x}, \ \ \forall x \in [N].
\ee
Notice, upto this point, we have not invoked \emph{any} specifics of operator list $\mathcal{O}$ except of it being composed of linear operators. Consequently, the aforementioned constraints are independent of the operator list. Let us now consider an operator list, $\mathcal{O}_1\equiv\{\I,\{\{M_{z|y}\}^{D-1}_{z=1}\}^M_{y=1}\}$. Immediately, the first entry of each moment matrix must be unity $(\Gamma_x)_{\I,\I}=1$, and some entries directly correspond to probabilities, $p(z|x,y)=\tr(\rho_x M^y_z)=(\Gamma^{\mathcal{O}_1}_{\rho_x})_{\I,M_{z|y}}$ (for all $z\in[D-1]$ and $p(z=D|x,y)=1-\sum^{D-1}_{z=1}(\Gamma^{\mathcal{O}_1}_{\rho_x})_{\I,M_{z|y}}$ for all $y\in [M], x\in[N]$). 

As the measurements and the dimension of the quantum systems remain unconstrained in communication complexity tasks, we can use the Naimark’s dilation theorem, and without loss of generality, take the measurements to be \emph{sharp} or protective, such that we have additional constraints $(\Gamma^{\mathcal{O}_1}_{\tau})_{M_{z'|y},M_{z|y}}=M_{z|y}^\dagger M_{z'|y}=\delta_{z,z'}M_{z|y}=(\Gamma^{\mathcal{O}_1}_{\tau})_{\I,M_{z|y}}$ (for all $z,z'\in[D-1],y\in[M],\tau\in\{\{\rho_x\}^N_{x=1},\Theta\}$). Consequently, we arrive at the following optimization problem,

\begin{widetext}
\begin{equation} \label{SvsDQ1}
\begin{split} \mathcal{S}^{max}_{Q_{1}}=
 \underset{\{\Gamma^{\mathcal{O}_1}_{\rho_x}\}^N_{x=1},\Gamma^{\mathcal{O}_1}_\Theta}{\ \ \max\ \ } & \sum_{x,y,z} c(x,y,z)p(z|x,y) \\
\text{s.t.\ \ } & \Gamma^{\mathcal{O}_1}_{\Theta} \geqslant \Gamma^{\mathcal{O}_1}_{\rho_x}, \ \ \forall \ x \in [N], \\
&  \frac1N \Gamma^{\mathcal{O}_1}_{\Theta} \leqslant p,\\
&  p(z|x,y) = (\Gamma^{\mathcal{O}_1}_{\rho_x})_{\I,M_{z|y}}, \ \ \forall \ z\in[D-1],y\in[M],x\in[N], \\
&  p(z=D|x,y) = 1-\sum^{D-1}_{z=1}(\Gamma^{\mathcal{O}_1}_{\rho_x})_{\I,M_{z|y}}, \ \ \forall \  y\in[M],x\in[N], \\
&  \Gamma^{\mathcal{O}_1}_{\rho_x}\geqslant 0, \ \ \forall \ x\in[N], \\
&  (\Gamma^{\mathcal{O}_1}_{\rho_x})_{\mathbb{I},\mathbb{I}}=1, \ \ \forall \ x\in[N], \\
&  (\Gamma^{\mathcal{O}_1}_{\tau})_{M_{z'|y},M_{z|y}}=(\Gamma^{\mathcal{O}_1}_{\tau})_{\I,M_{z|y}}, \ \  \forall \ y\in[M], z,z'\in[D-1],\tau\in\{\{\rho_x\}^N_{x=1},\Theta\}.
\end{split}
\end{equation}
\end{widetext}
As all constraints in \eqref{SvsDQ1} are satisfied by quantum protocols, the optimization problem \eqref{SvsDQ1} forms a dimension-independent relaxation of \eqref{SvsDQSeeSaw}. Consequently, we retrieve an upper bound on the absolute maximum quantum value of the success metric, $\mathcal{S}^{max}_{Q_{1}}\geqslant \mathcal{S}^{max}_{Q}\geqslant \mathcal{S}^{max}_{Q_{LB}}$. Now, $Q_1$ with the operator list $\mathcal{O}_1$ is the first level relaxation in a hierarchy of SDP relaxations of the optimization problem \eqref{SvsDQ}. We now define $Q_{\mathcal{L}}$, where $\mathcal{L}\in \mathbb{N}$, as the $\mathcal{L}$th level of the relaxation, associated with the operator list $\mathcal{O}_{\mathcal{L}}$ containing all monomials of operators contained in $\mathcal{O}_1$ of length at-most $\mathcal{L}$. As $\mathcal{O}_1\subseteq \mathcal{O}_{\mathcal{L}}$, for $\mathcal{L}\geqslant 1$, the hierarchy retrieves a series of tightening upper-bounds $S^{max}_{Q_{\mathcal{L}}}$ on the maximum success metric given $\mathcal{D}_Q \leq p$, such that, $\mathcal{S}^{max}_{Q_{1}}\geqslant \mathcal{S}^{max}_{Q_{\mathcal{L}}}\geqslant \mathcal{S}^{max}_{Q_{\mathcal{L}+1}}\geqslant \mathcal{S}^{max}_{Q}\geqslant \mathcal{S}^{max}_{Q_{LB}}$. Whenever an upper bound from the hierarchy coincides (up to numerical precision) with the lower bound from the see-saw method, we retrieve the maximum quantum value of the success metric, $\mathcal{S}^{max}_{Q_{\mathcal{L}}}= \mathcal{S}^{max}_{Q}= \mathcal{S}^{max}_{Q_{LB}}$, given $\mathcal{D}_Q \leq p$, for any $\mathcal{L}\geqslant 1$. This hierarchy also  retrieves a series of tightening lower bounds $\mathcal{D}_{Q_{\mathcal{L}}}$ on the minimal distinguishability required to achieve the success metric $\mathcal{S}_{Q} = \mathcal{S}^{max}_{Q_{\mathcal{L}}}$, such that $\mathcal{D}_{Q_{1}}\leqslant \mathcal{D}_{Q_{\mathcal{L}}}\leqslant \mathcal{D}_{Q_{\mathcal{L}+1}}\leqslant \mathcal{D}_{Q}\leqslant \mathcal{D}_{Q_{UB}}$. Whenever,  $\mathcal{D}_{Q_{UB}}=\mathcal{D}_{Q_{\mathcal{L}}}$ (up to machine precision), we retrieve the absolute minimal quantum distinguishability required to achieve a given value of the success metric $\mathcal{S}_{Q} = \mathcal{S}^{max}_{Q_{\mathcal{L}}}$. We use these methods to obtain Figure \ref{fig:rac}, which displays a thorough investigation of quantum advantage for RACs with parameters $(n=2, d=3)$, $(n=3, d=2)$.
 
Alternatively, following a similar methodology, we can directly formulate a see-saw SDP method for obtaining dimension dependent lower-bounds, as a well as hierarchy of SDP relaxations to retrieve dimension independent lower bounds on the distinguishability of quantum communication given a lower bound on the success metric of a communication. Clearly, the optimal solutions of this optimization problem will coincide with optimal solutions of the optimization problem \eqref{SvsDQ}.

\end{document}